\documentclass[10pt,twocolumn,twoside,aps,pra,nopacs,superscriptaddress]{revtex4-1}
\usepackage{graphicx}
\usepackage{amsmath,amssymb}
\usepackage{color}
\usepackage{bbm}
\usepackage{hyperref}
\usepackage{theorem}
\usepackage{latexsym}
\usepackage{comment}

\newtheorem{theorem}{Theorem}

\newtheorem{corollary}[theorem]{Corollary}

\newtheorem{lemma}[theorem]{Lemma}
\newtheorem{meta-lemma}[theorem]{Meta-Lemma}

\newtheorem{proposition}[theorem]{Proposition}
\newtheorem{remark}[theorem]{Remark}

\newlength{\blank}
\newenvironment{proof}{{\noindent\textbf{Proof.\ }}}{\hfill\qed}
\newenvironment{proofthm}[1]{{\noindent\textbf{Proof~{#1}.\ }}}{\hfill\qed}
\newcommand{\ket}[1]{|#1\rangle}
\newcommand{\bra}[1]{\langle#1|}

\mathchardef\ordinarycolon\mathcode`\:
\mathcode`\:=\string"8000
\def\vcentcolon{\mathrel{\mathop\ordinarycolon}}
\begingroup \catcode`\:=\active
  \lowercase{\endgroup
  \let :\vcentcolon
  }

\newcommand{\nc}{\newcommand}
\nc{\rnc}{\renewcommand}
\nc{\beq}{\begin{equation}}
\nc{\eeq}{{\end{equation}}}
\nc{\beqa}{\begin{eqnarray}}
\nc{\eeqa}{\end{eqnarray}}
\nc{\lbar}[1]{\overline{#1}}
\nc{\ketbra}[2]{|#1\rangle\!\langle#2|}
\nc{\proj}[1]{| #1\rangle\!\langle #1 |}
\nc{\avg}[1]{\langle#1\rangle}
\nc{\Rank}{\operatorname{Rank}}
\nc{\smfrac}[2]{\mbox{$\frac{#1}{#2}$}}
\nc{\tr}{\operatorname{Tr}}
\nc{\ox}{\otimes}
\nc{\dg}{\dagger}
\nc{\dn}{\downarrow}
\nc{\cA}{\mathcal{A}}
\nc{\cB}{\mathcal{B}}
\nc{\cC}{\mathcal{C}}
\nc{\cD}{\mathcal{D}}
\nc{\cE}{\mathcal{E}}
\nc{\cF}{\mathcal{F}}
\nc{\cG}{\mathcal{G}}
\nc{\cH}{\mathcal{H}}
\nc{\cI}{\mathcal{I}}
\nc{\cJ}{\mathcal{J}}
\nc{\cK}{\mathcal{K}}
\nc{\cL}{\mathcal{L}}
\nc{\cM}{\mathcal{M}}
\nc{\cN}{\mathcal{N}}
\nc{\cO}{\mathcal{O}}
\nc{\cP}{\mathcal{P}}
\nc{\cR}{\mathcal{R}}
\nc{\cS}{\mathcal{S}}
\nc{\cT}{\mathcal{T}}
\nc{\cX}{\mathcal{X}}
\nc{\cZ}{\mathcal{Z}}
\nc{\csupp}{{\operatorname{csupp}}}
\nc{\qsupp}{{\operatorname{qsupp}}}
\nc{\var}{\operatorname{var}}
\nc{\rar}{\rightarrow}
\nc{\lrar}{\longrightarrow}
\nc{\polylog}{\operatorname{polylog}}
\nc{\1}{{\openone}}
\nc{\id}{{\operatorname{id}}}

\nc{\RR}{{{\mathbb R}}}
\nc{\CC}{{{\mathbb C}}}
\nc{\FF}{{{\mathbb F}}}
\nc{\NN}{{{\mathbb N}}}
\nc{\ZZ}{{{\mathbb Z}}}
\nc{\PP}{{{\mathbb P}}}
\nc{\QQ}{{{\mathbb Q}}}
\nc{\UU}{{{\mathbb U}}}
\nc{\EE}{{{\mathbb E}}}

\nc{\qed}{{\hfill$\Box$}}

\def\>{\rangle}
\def\<{\langle}

\begin{document}

\title{Tight uniform continuity bounds for quantum entropies: \protect\\
       conditional entropy, relative entropy distance and energy constraints}

\author{Andreas Winter}
\email{andreas.winter@uab.cat}
\affiliation{ICREA \& F\'{\i}sica Te\`{o}rica: Informaci\'{o} i Fen\`{o}mens Qu\`{a}ntics, %
Universitat Aut\`{o}noma de Barcelona, ES-08193 Bellaterra (Barcelona), Spain}

\date{12 January 2016}

\begin{abstract}
  We present a bouquet of continuity bounds for quantum entropies, falling
  broadly into two classes:
  First, a tight analysis of the Alicki-Fannes continuity bounds
  for the conditional von Neumann entropy, reaching almost the best possible 
  form that depends only on the system dimension and the trace distance
  of the states. Almost the same proof can be used to derive similar
  continuity bounds for the relative entropy distance from a convex set 
  of states or positive operators. 
  As applications we give new proofs, with tighter bounds, of the asymptotic continuity of 
  the relative entropy of entanglement, $E_R$, and its regularization $E_R^\infty$,
  as well as of the entanglement of formation, $E_F$. Using a novel
  ``quantum coupling'' of density operators, which may be of independent
  interest, we extend the latter to an asymptotic continuity bound for the 
  regularized entanglement of formation, aka entanglement cost, $E_C=E_F^\infty$.

  Second, we derive analogous continuity bounds for the von Neumann entropy and
  conditional entropy in infinite dimensional systems under an energy
  constraint, most importantly systems of multiple quantum harmonic
  oscillators. While without an energy bound the entropy is discontinuous, it is
  well-known to be continuous on states of bounded energy. However, a 
  quantitative statement to that effect seems not to have been known. Here, under 
  some regularity assumptions on the Hamiltonian, we find
  that, quite intuitively, the Gibbs entropy at the given energy
  roughly takes the role of the Hilbert space dimension in the finite-dimensional 
  Fannes inequality.
\end{abstract}


\maketitle

\section{Introduction}
On finite dimensional systems, the von Neumann entropy 
$S(\rho) = -\tr\rho\log\rho$ is continuous,
but this becomes useful only once one has explicit continuity bounds,
most significantly the one due to Fannes~\cite{Fannes}, the sharpest
form of which is the following:
\begin{lemma}[{Audenaert~\cite{Audenaert:Fannes}, Petz~\cite{Petz:book}}]
  \label{lemma:fannes-audenaert}
  For states $\rho$ and $\sigma$ on a Hilbert space $A$ of dimension
  $d = |A| < \infty$, if $\frac12\|\rho-\sigma\|_1 \leq \epsilon \leq 1$, then
  \[
    |S(\rho)-S(\sigma)| \leq \begin{cases}
                               \epsilon \log(d-1) + h(\epsilon) & \text{ if } \epsilon \leq 1-\frac1d, \\
                               \log d                           & \text{ if } \epsilon > 1-\frac1d,
                             \end{cases}
  \]
  with $h(x) = H(x,1-x) = -x\log x-(1-x)\log(1-x)$ the binary entropy.
  A simplified, but universal bound reads
  \[
    |S(\rho)-S(\sigma)| \leq \epsilon \log d + h(\epsilon).
  \]
\end{lemma}
We include a short proof for self-containedness, and also because 
it deserves to be known better. It seems that it was first found by 
Petz~\cite[Thm.~3.8]{Petz:book}, who credits 
Csisz\'{a}r for the classical case; the latter seems to have appeared
first in Zhang's paper~\cite{Zhang} (see also~\cite{Sason}).

\medskip
\begin{proof}
We only have to treat the case $\epsilon \leq 1-\frac1d$.
We begin with the classical case of two probability distributions
$p$ and $q$ on the same ground set of $d$ elements.  
It is well known, and in fact elementary to confirm, 
that one can find two jointly distributed random variables, $X \sim p$ 
and $Y\sim q$
(meaning $X$ is distributed according to the probability law $p$,
and $Y$ according to $q$), 
with $\Pr\{X\neq Y\} = \frac12 \|p-q\|_1 \leq \epsilon$.
The crucial idea is to let $\Pr\{X=Y=x\} = \min(p_x,q_x)$ and to distribute
the remaining probability weight suitably off the diagonal.
(This is also the minimum probability over all such coupled
random variables~\cite{Zhang}.
For the reader with a taste for the sophisticated,
this is the Kantorovich-Rubinshtein dual formula for the Wasserstein 
distance in the case of the trivial
metric $d(x,y)=1$ for all $x\neq y$ and $d(x,x)=0$,
cf.~the broad survey~\cite{Wasserstein-survey}.)
Then, by the monotonicity of the Shannon entropy under taking marginals
and Fano's inequality (see~\cite{CoverThomas}),
\[\begin{split}
  H(X)-H(Y) &\leq H(XY)-H(Y) \\
            &=    H(X|Y) \leq \epsilon \log(d-1) + h(\epsilon),
\end{split}\]
and likewise for $H(Y)-H(X)$. 
[For the simplified bound, we use $H(X|Y) \leq \epsilon \log d + h(\epsilon)$.]

Next, we reduce the quantum case to the classical one: 
W.l.o.g.~$S(\rho) \leq S(\sigma)$,
and consider the dephasing operation $E$ in the eigenbasis of $\rho$, 
which maps $\rho$ to itself, a diagonal matrix with a probability distribution
$p$ along the diagonal, and $\sigma$ to $E(\sigma)$, 
a diagonal matrix with a probability distribution $q$ along the diagonal.
Hence
\[
  H(p) = S(\rho) \leq S(\sigma) \leq S\bigl(E(\sigma)\bigr) = H(q).
\]
At the same time, $\|p-q\|_1 = \| E(\rho)-E(\sigma) \|_1 \leq \|\rho-\sigma\|_1$,
and so, using the classical case,
\[
  |S(\rho)-S(\sigma)| \leq H(q)-H(p) \leq \epsilon\log(d-1) + h(\epsilon).
\]

Note that the inequality is tight for all $\epsilon$ and $d$, 
e.g.~by $\sigma = \proj{0}$
and $\rho = (1-\epsilon)\proj{0} + \frac{\epsilon}{d-1}(\1-\proj{0})$.
\end{proof}

\medskip
We are interested in bounds of the above form, i.e.~only referring to the trace
distance of the states and some general global parameter specifying the
system, for a number of entropic quantities, starting with the conditional
von Neumann entropy, relative entropy distances from certain sets, etc, which
have numerous applications in quantum information theory and quantum statistical
physics. Furthermore, and perhaps even more urgently, in situations of
infinite dimensional Hilbert spaces, where the above form of the Fannes
inequality becomes trivial.

\medskip
The rest of the paper is structured as follows: in Section~\ref{sec:conditional}
we present and prove an almost tight version of Lemma~\ref{lemma:fannes-audenaert}
for the conditional entropy (originally due to Alicki and Fannes~\cite{AlickiFannes}),
then in Section~\ref{sec:relative} we generalize the principle behind our proof 
to a family of relative entropy distance measures from a convex set; in these two
sections we also present some illustrative applications of the conditional
entropy bounds to two entanglement measures, $E_R$ and $E_F$, as well as their
regularizations.
In Section~\ref{sec:bounded} we expand the methodology of the first part of the
paper to infinite dimensional systems, where Fannes-type continuity bounds
are obtained under an energy constraint for a broad class of Hamiltonians,
and specifically for quantum harmonic oscillators. 
All entropy continuity bounds are stated as \emph{Lemmas}, while the 
applications appear as \emph{Corollaries}, and two auxiliary
results (on ``quantum coupling'' of density matrices) as
\emph{Propositions}. The absence of \emph{Theorems} is meant
to encourage readers to apply the results presented here.

\section{Conditional entropy}
\label{sec:conditional}
Alicki and Fannes~\cite{AlickiFannes} proved an extension of the 
Fannes inequality for the conditional entropy
\[
  S(A|B)_\rho = S(\rho^{AB})-S(\rho^B),
\]
defined for states $\rho$ on a bipartite (tensor product) Hilbert space
$A\otimes B$. While a double application of Lemma~\ref{lemma:fannes-audenaert}
would yield such a bound involving both the dimensions of $A$ and $B$,
Alicki and Fannes show that if $\|\rho-\sigma\|_1 \leq \epsilon \leq 1$,
then
\[
  \bigl|S(A|B)_\rho-S(A|B)_\sigma\bigr| \leq 4\epsilon \log|A| + 2h(\epsilon).
\]
In particular, this form is independent of the dimension of $B$,
which might even be infinite. Note that for classical, Shannon,
conditional entropy, an inequality like the above can be obtained from 
Lemma~\ref{lemma:fannes-audenaert} by convex combination, resulting
in a bound like that of Lemma~\ref{lemma:fannes-audenaert} (see below).

The Alicki-Fannes inequality has several applications in quantum 
information theory, from the proof of asymptotic continuity of
entanglement measures --- most notably squashed entanglement~\cite{E-sq} and 
conditional entanglement of mutual information (CEMI)~\cite{YHW} ---,
to the continuity of quantum channel capacities~\cite{LeungSmith}, and 
on to the recent discussion of approximately degradable channels~\cite{Sutter-et-al}.

We present a simple proof of the Alicki-Fannes inequality that yields the
stronger form of Lemma~\ref{lemma:alicki-fannes-new}.
One of the themes of the present paper, to which we draw attention here, 
is the use of entropy inequalities in the proofs. 
In particular, we make use of the concavity of the conditional 
entropy (which is equivalent to strong subadditivity of the von Neumann 
entropy)~\cite{SSA}. 
In the following proof we will specifically rely on two inequalities
expressing the concavity of the entropy and the fact that it is not 
``too concave''~\cite{KimRuskai}:
\begin{equation}
  \label{eq:concavity-bounds}
  \sum_i p_i S(\rho_i) \leq S\left( \sum_i p_i\rho_i \right) \leq \sum_i p_i S(\rho_i) + H(p).
\end{equation}
By introducing a bipartite state $\rho = \sum_i p_i \rho_i^A \ox \proj{i}^I$,
this is seen to be equivalent to
\[
  S(A|I) \leq S(A) \leq S(AI) = S(A|I) + S(I),
\]
which consists of two applications of strong subadditivity.

\begin{lemma}
  \label{lemma:alicki-fannes-new}
  For states $\rho$ and $\sigma$ on a Hilbert space $A\otimes B$,
  if $\frac12\|\rho-\sigma\|_1 \leq \epsilon \leq 1$, then
  \[
    \bigl|S(A|B)_\rho - S(A|B)_\sigma\bigr| 
         \leq 2\epsilon \log|A| + (1+\epsilon)\,h\!\left(\!\frac{\epsilon}{1+\epsilon}\!\right).
  \]
  
  If $B$ is classical in the sense that both $\rho$ and $\sigma$ are so-called
  qc-states, i.e.~with an orthonormal basis $\{\ket{x}\}$,
  \[
    \rho   = \sum_x p_x \rho_x^A \ox \proj{x}^B, \quad
    \sigma = \sum_x q_x \sigma_x^A \ox \proj{x}^B,
  \]
  and analogously if both are cq-states,
  then this can be tightened to
  \[
    \bigl|S(A|B)_\rho - S(A|B)_\sigma\bigr| 
          \leq \epsilon \log|A| + (1+\epsilon)\,h\!\left(\!\frac{\epsilon}{1+\epsilon}\!\right).
  \]
\end{lemma}
\begin{proof}
The right hand side is monotonic in $\epsilon$, hence we may assume
$\frac12\|\rho-\sigma\|_1 = \epsilon$.
Let $\epsilon\Delta = (\rho-\sigma)_+$ be the positive part of $\rho-\sigma$.
Note that because this difference is traceless and its trace
norm equals $2\epsilon$, $\Delta$ is a bona fide state.
Furthermore,
\[\begin{split}
  \rho &=    \sigma + (\rho-\sigma) \\
       &\leq \sigma + \epsilon\Delta \\
       &=    (1+\epsilon)\left( \frac{1}{1+\epsilon}\sigma 
                               +\frac{\epsilon}{1+\epsilon}\Delta \right) \\
       &=:   (1+\epsilon)\omega.
\end{split}\]
By letting $\epsilon\Delta' := (1+\epsilon)\omega-\rho$, we obtain
another state $\Delta'$, such that
\begin{equation}
  \label{eq:decompositions}
  \omega = \frac{1}{1+\epsilon}\sigma + \frac{\epsilon}{1+\epsilon}\Delta
         = \frac{1}{1+\epsilon}\rho + \frac{\epsilon}{1+\epsilon}\Delta'.
\end{equation}
This is a slightly optimized version of the trick in the
proof of Alicki and Fannes~\cite{AlickiFannes}; cf.~\cite{MosonyiHiai}.

Now, we use the following well-known variational characterization of
the conditional entropy:
\[
  -S(A|B)_\omega = \min_\xi D\bigl(\omega^{AB}\|\1^A\ox\xi^B\bigr),
\]
where $D(\rho\|\sigma) = \tr\rho(\log\rho-\log\sigma)$ is the
quantum relative entropy~\cite{Umegaki,OhyaPetz}.
Choosing an optimal state $\xi$ for $\omega$ (which is $\xi = \omega^B$),
we have, from Eq.~(\ref{eq:decompositions}),
\[\begin{split}
  S(A|B)_\omega &=   -D\bigl(\omega^{AB}\|\1^A\ox\xi^B\bigr)         \\
                &=    S(\omega) + \tr\omega \log\xi^B                \\
                &\leq h\!\left(\!\frac{\epsilon}{1+\epsilon}\!\right)
                       + \frac{1}{1+\epsilon}S(\rho) + \frac{\epsilon}{1+\epsilon}S(\Delta') \\
                &\phantom{h\!\left(\!\frac{\epsilon}{1+\epsilon}\!\right)}
                       + \frac{1}{1+\epsilon}\tr\rho\log\xi^B + \frac{\epsilon}{1+\epsilon}\tr\Delta'\log\xi^B \\
                &=    h\!\left(\!\frac{\epsilon}{1+\epsilon}\!\right)
                       \!-\! \frac{1}{1\!+\!\epsilon}D(\rho\|\1\ox\xi) 
                       \!-\! \frac{\epsilon}{1\!+\!\epsilon}D(\Delta'\|\1\ox\xi) \\
                &\leq h\!\left(\!\frac{\epsilon}{1+\epsilon}\!\right)
                       + \frac{1}{1+\epsilon} S(A|B)_\rho + \frac{\epsilon}{1+\epsilon} S(A|B)_{\Delta'},
\end{split}\]
where in the third line we have used the concavity upper bound
from Eq.~(\ref{eq:concavity-bounds}).
Using the other decomposition in Eq.~(\ref{eq:decompositions}),
the concavity of the conditional entropy, i.e.~the lower bound
in Eq.~(\ref{eq:concavity-bounds}), gives
\[
  S(A|B)_\omega \geq \frac{1}{1+\epsilon} S(A|B)_\sigma + \frac{\epsilon}{1+\epsilon} S(A|B)_{\Delta}.
\]

Putting these two bounds together and multiplying by $1+\epsilon$, we
arrive at
\[\begin{split}
  S(A|B)_\sigma - S(A|B)_\rho &\leq \epsilon\bigl( S(A|B)_{\Delta'}-S(A|B)_{\Delta}\bigr) \\
                              &\phantom{=:} + (1+\epsilon) \,h\!\left(\!\frac{\epsilon}{1+\epsilon}\!\right).
\end{split}\]
The proof of the general bound is concluded observing that the conditional 
entropy of any state is bounded between $-\log|A|$ and $+\log|A|$.

\medskip
For the case of two qc-states or two cq-states as above, note that the
states $\Delta$ and $\Delta'$ are of the same, qc-form (cq-form, resp.),
and so their conditional entropies are between
$0$ and $\log|A|$.
\end{proof}

\begin{remark}
\normalfont
Lemma~\ref{lemma:alicki-fannes-new} is almost best possible, as we can
see by considering the example of $\sigma^{AB} = \Phi_d$, the maximally
entangled state on $A=B=\CC^d$, and 
$\rho^{AB} = (1-\epsilon) \Phi_d + \frac{\epsilon}{d^2-1}(\1-\Phi_d)$.
Clearly, $\frac12 \|\rho-\sigma\|_1 = \epsilon$, while
\[\begin{split}
  S(A|B)_\rho & -S(A|B)_\sigma \\
              &= \bigl( \epsilon\log(d^2-1)+h(\epsilon)-\log d\bigr)-(-\log d) \\
              &= 2\epsilon\log d + h(\epsilon) - O\left(\frac{\epsilon}{d^2}\right).
\end{split}\]
This asymptotically matches Lemma~\ref{lemma:alicki-fannes-new}
for large $d$ and small $\epsilon$.
\end{remark}

\medskip
As an application of Lemma~\ref{lemma:alicki-fannes-new}, we 
can obtain tighter continuity bounds on various quantum channel 
capacities, simply substituting our tighter bound rather than the original formulation
of Alicki and Fannes in the proofs of Leung and Smith~\cite{LeungSmith}.

As a token, we demonstrate a tight version of the asymptotic continuity 
of the entanglement of formation~\cite{BDSW},
\[
  E_F(\rho) = \inf \sum_x p_x S(\tr_B \rho_x) \text{ s.t. } \rho = \sum_x p_x\rho_x
\]
for a state $\rho^{AB}$ on the bipartite system $A\ox B$,
originally due to Nielsen~\cite{Nielsen-continuity}.
We then go on to prove asymptotic continuity for its regularization, 
the entanglement cost~\cite{HHT-EC},
\[
  E_C(\rho) = E_F^\infty(\rho) = \lim_{n\rightarrow\infty} \frac1n E_F(\rho^{\ox n}),
\]
which, albeit following the general ``telescoping''
strategy of~\cite{LeungSmith}, requires a new idea, and seems 
not to have been known before~\cite{antisymm}. 
Note that $E_C$ is different from $E_F$~\cite{Hastings}.

\begin{corollary}
  \label{cor:EoF}
  Let $\rho$ and $\sigma$ be states on the system $A\ox B$, denoting the smaller of 
  the two dimensions by $d$. Then, $\frac12 \|\rho-\sigma\|_1 \leq \epsilon$ implies,
  with $\delta = \sqrt{\epsilon(2-\epsilon)}$,
  \begin{align*}
    |E_F(\rho)-E_F(\sigma)| &\leq   \delta \log d + (1+\delta)\,h\!\left(\!\frac{\delta}{1+\delta}\!\right), \\
    |E_C(\rho)-E_C(\sigma)| &\leq 2 \delta \log d + (1+\delta)\,h\!\left(\!\frac{\delta}{1+\delta}\!\right).
  \end{align*}
\end{corollary}
Note that these bounds only depend on the smaller of the two
dimensions, in contrast to~\cite{Nielsen-continuity}; in particular,
they apply even in the case that one of the two Hilbert spaces is infinite
dimensional.

\medskip
\begin{proof}
We may assume w.l.o.g.~that $E_F(\rho) \geq E_F(\sigma)$ and $|B| \geq |A| = d$. 
Choose a purifying system $R \simeq AB$,
and pure states $\varphi^{ABR}$ and $\psi^{ABR}$ with 
$\varphi^{AB}=\rho$ and $\psi^{AB}=\sigma=\psi^R$ such that
\[
  |\bra{\varphi}\psi\rangle| = F(\rho,\sigma) \geq 1-\epsilon,
\]
thus $\frac12 \|\varphi-\psi\|_1 \leq \delta = \sqrt{1-(1-\epsilon)^2}$.
Here, $F(\rho,\sigma) = \|\sqrt{\rho}\sqrt{\sigma}\|_1$ is the fidelity
between two quantum states, and we have used that it is related to the 
trace distance by these well-known inequalities~\cite{FvdG}:
\begin{equation}
  \label{eq:FvdG}
  1-F(\rho,\sigma) \leq \frac12 \|\rho-\sigma\|_1 \leq \sqrt{1-F(\rho,\sigma)^2}.
\end{equation}

By an observation of Schr\"odinger (which he called ``steering'') in the 
context of his investigation of quantum entanglement~\cite{Schroedinger:steering}, 
cf.~\cite{HughstonJozsaWootters}, for any convex decomposition 
$\sigma = \sum_x p_x \sigma_x$, there exists a measurement POVM $(M_x)$ 
on $R$ such that $p_x \sigma_x = \tr_R \psi(\1^{AB}\ox M_x^R)$. Introducing
the qc-channel $\cM(\xi) = \sum_x \tr\xi M_x \proj{x}$ from $R$ to a suitable
space $X$, we then have
\begin{align}
  \label{eq:sigma-tilde}
  \widetilde{\sigma} := (\id_{AB}\ox\cM)\psi
                              &= \sum_x p_x \sigma_x^{AB} \ox \proj{x}^X, \\
  \text{and}\quad 
  S(A|X)_{\widetilde{\sigma}} &= \sum_x p_x S(\tr_B \sigma_x). \nonumber
\end{align}
Let us choose an optimal decomposition for the purpose of entanglement
of formation, and the corresponding POVM and quantum channel, 
i.e.~$E_F(\sigma) = S(A|X)_{\widetilde{\sigma}}$. Applying the
same to $\varphi^{ABR}$, we obtain
\[
  \widetilde{\rho} := (\id_{AB}\ox\cM)\varphi
                    = \sum_x q_x \rho_x^{AB} \ox \proj{x}^X, 
\]
with $q_x = \tr\varphi^R M_x$. Hence,
\[
  E_F(\rho) \leq \sum_x p_x S(\tr_B \rho_x) = S(A|X)_{\widetilde{\rho}}.
\]
Observe that by the contractivity of the trace norm under cptp maps,
\[
  \delta \geq \|\psi-\varphi\|_1 
         \geq \|\widetilde{\sigma}-\widetilde{\rho}|\|_1.
\]
Now we can invoke the classical part of Lemma~\ref{lemma:alicki-fannes-new},
\[\begin{split}
  E_F(\rho) - E_F(\sigma) &\leq S(A|X)_{\widetilde{\rho}} - S(A|X)_{\widetilde{\sigma}} \\
                          &\leq \delta \log d + (1+\delta)\,h\!\left(\!\frac{\delta}{1+\delta}\!\right),
\end{split}\]
and we are done.

For the regularization, consider any integer $n$ and
\begin{equation}\begin{split}
  \label{eq:chain-trick}
  \Big| E_F&\bigl(\rho^{\ox n}\bigr) - E_F\bigl(\sigma^{\ox n}\bigr) \Big| \\
       &=     \left| \sum_{t=1}^n E_F\bigl(\rho^{\ox t} \ox \sigma^{\ox n-t}\bigr) 
                                  - E_F\bigl(\rho^{\ox t-1} \ox \sigma^{\ox n-t+1}\bigr) \right| \\
       &\leq \sum_{t=1}^n | E_F(\rho \ox \Omega_t) - E_F(\sigma \ox \Omega_t) |,
\end{split}\end{equation}
with $\Omega_t = \rho^{\ox t-1} \ox \sigma^{\ox n-t}$. The proof will be
concluded by showing that for any $\Omega^{A'B'}$,
\[
  | E_F(\rho \ox \Omega) - E_F(\sigma \ox \Omega) | 
                 \leq 2\delta\log d + (1+\delta) \,h\!\left(\!\frac{\delta}{1+\delta}\!\right),
\]
as this will imply from Eq.~(\ref{eq:chain-trick}) that 
\[
  \frac1n \Big| E_F\bigl(\rho^{\ox n}\bigr) - E_F\bigl(\sigma^{\ox n}\bigr) \Big| 
                 \leq 2\delta\log d + (1+\delta) \,h\!\left(\!\frac{\delta}{1+\delta}\!\right).
\]
To see this, assume again w.l.o.g.~that
$E_F(\rho \ox \Omega) \geq E_F(\sigma \ox \Omega)$, and choose a purification
$\upsilon$ of $\Omega$ on $A'B'R'$, with $R' \simeq A'B'$. 
Besides the purification $\psi^{ABR}$ of $\sigma$, we now
need a state (not generally pure) $\Theta^{ABR}$ with
$\Theta^{AB}=\rho$ and $\Theta^R=\psi^R$. Proposition~\ref{prop:quantum-coupling}
below guarantees the existence of such a state with
$F(\psi,\Theta) \geq 1-\epsilon$, hence 
$\frac12 \| \psi-\Theta \|_1 \leq \delta$, once more invoking Eq.~(\ref{eq:FvdG}).
As before we choose an optimal decomposition of $\sigma^{AB}\ox\Omega^{A'B'}$
into states on $AA':BB'$, which we can represent by a POVM and associated
cptp map $\cM:RR' \longrightarrow X$:
\begin{align*}
  \widetilde{\sigma}   &:= (\id_{AA'BB'}\ox\cM)(\psi\ox\upsilon) \\
                       &= \sum_x p_x \sigma_x^{AA'BB'} \ox \proj{x}^X, \\
  E_F(\sigma\ox\Omega) &= S(AA'|X)_{\widetilde{\sigma}} 
                        = \sum_x p_x S(\tr_{BB'} \sigma_x).
\end{align*}
Applying the same map to $\omega\ox\upsilon$, we get
\begin{align*}
  \widetilde{\rho}     &:= (\id_{AA'BB'}\ox\cM)(\Theta\ox\upsilon) \\
                       &= \sum_x p_x \rho_x^{AA'BB'} \ox \proj{x}^X, \\
  E_F(\sigma\ox\Omega) &\leq S(AA'|X)_{\widetilde{\rho}} 
                        =    \sum_x p_x S(\tr_{BB'} \rho_x),
\end{align*}
where we observe that, crucially, the same $p_x$ appear in the 
expressions for $\widetilde{\rho}$ and $\widetilde{\sigma}$. 
Using $\Theta^R = \sigma^T = \psi^R$, we even have
\[
  \widetilde{\rho}^{A'B'X} = (\id_{A'B'}\ox\cM)(\sigma^T\ox\upsilon)
                           = \widetilde{\sigma}^{A'B'X}.
\]
Thus with Lemma~\ref{lemma:alicki-fannes-new}, as desired,
\[\begin{split}
  E_F(\rho\ox\Omega) - E_F(\sigma\ox\Omega)
                     &\leq S(AA'|X)_{\widetilde{\rho}} - S(AA'|X)_{\widetilde{\sigma}} \\
                     &=    S(A|A'X)_{\widetilde{\rho}} - S(A|A'X)_{\widetilde{\sigma}} \\ 
                     &\leq 2\delta\log d + (1+\delta)\,h\!\left(\!\frac{\delta}{1+\delta}\!\right),
\end{split}\]
where in the second line we have used the chain rule
$S(AA'|X) = S(A'|X) + S(A|A'X)$, as well as 
$S(A'|X)_{\widetilde{\rho}} = S(A'|X)_{\widetilde{\sigma}}$.
\end{proof}


\begin{proposition}[``Quantum coupling'']
  \label{prop:quantum-coupling}
  Given states $\rho$ and $\sigma$ on a Hilbert space $A$,
  with $\frac12 \|\rho-\sigma\|_1 \leq \epsilon$,
  there exist purifications $\ket{\varphi}$ of $\rho$ and 
  $\ket{\psi}$ of $\sigma$, and a (sub-normalized) vector $\ket{\vartheta}$, 
  all three in the tensor square Hilbert space $A \ox A =:A_1 A_2$, such that
  \begin{align*}
    \rho^T\! = \varphi^{A_2}, &\quad \rho       = \varphi^{A_1} \geq \tr_{A_2} \proj{\vartheta}, \\ 
    \sigma   = \psi^{A_1},    &\quad \sigma^T\! = \psi^{A_2}    \geq \tr_{A_1} \proj{\vartheta},
  \end{align*}
  and
  \[
    |\bra{\psi} \vartheta\rangle|,\ |\bra{\varphi} \vartheta\rangle| \geq 1-\epsilon.
  \]
  Here, ${\cdot}^T$ denotes the transpose of a matrix with respect to a chosen basis.

  Consequently, there exists a state $\Theta^{A_1A_2}$ with the properties
  $\Theta^{A_1} = \rho$ and $\Theta^{A_2}=\psi^{A_2}=\sigma^T$, and such that
  $F(\psi,\Theta),\,F(\varphi,\Theta) \geq 1-\epsilon$.
\end{proposition}

This proposition can be viewed as a quantum analogue of the 
coupling of random variables $X \sim p$ and $Y \sim q$ such that
$\Pr\{X\neq Y\} = \frac12 \|p-q\|_1$, on which the proof
of Lemma~\ref{lemma:fannes-audenaert} relied.

\medskip
\begin{proof}
Fixing an orthonormal basis $\{\ket{i}\}$ of $A$, and introducing the
unnormalized maximally entangled vector
\[
  \ket{\Phi} = \sum_i \ket{i}^{A_1}\ket{i}^{A_2},
\]
we have the following two ``pretty good purifications''~\cite{Winterize-or-die}
of $\rho$ and $\sigma$:
\begin{align*}
  \ket{\varphi} &:= (\sqrt{\rho}\ox\1)\ket{\Phi} = \left(\1\ox\sqrt{\rho}^T\right)\ket{\Phi}, \\
  \ket{\psi}    &:= (\sqrt{\sigma}\ox\1)\ket{\Phi} = \left(\1\ox\sqrt{\sigma}^T\right)\ket{\Phi},
\end{align*}
the claimed properties of which can be readily checked.

To obtain $\ket{\vartheta}$, we use once more Eq.~(\ref{eq:decompositions})
from the proof of Lemma~\ref{lemma:alicki-fannes-new}:
\begin{equation*}
  \omega = \frac{1}{1+\epsilon}\sigma + \frac{\epsilon}{1+\epsilon}\Delta
         = \frac{1}{1+\epsilon}\rho + \frac{\epsilon}{1+\epsilon}\Delta',
\end{equation*}
with states $\Delta$ and $\Delta'$. Then define
\[\begin{split}
  \ket{\vartheta} &:= \left( \frac{1}{\sqrt{1+\epsilon}}\rho^{1/2}\omega^{-1/2}\sigma^{1/2} \ox \1 \right)\ket{\Phi} \\
                  &=  \left( X \ox \sqrt{\sigma}^T \right)\ket{\Phi} 
                   =  \left( X \ox \1 \right)\ket{\psi}             \\
                  &=  \left( \sqrt{\rho} \ox Y \right)\ket{\Phi} 
                   =  \left( \1 \ox Y \right)\ket{\varphi},
\end{split}\]
using $(Z\ox\1)\ket{\Phi} = (\1\ox Z^T)\ket{\Phi}$, with
\begin{align*}
  X &= \frac{1}{\sqrt{1+\epsilon}}\rho^{1/2}\omega^{-1/2},         \\
  Y &= \frac{1}{\sqrt{1+\epsilon}}(\sigma^T)^{1/2}(\omega^T)^{-1/2}.
\end{align*}

We claim that $\|X\|,\,\|Y\| \leq 1$. Indeed, $\omega \geq \frac{1}{1+\epsilon}\rho$, so 
\[\begin{split}
  X X^\dagger &=    \frac{1}{1+\epsilon} \sqrt{\rho} \omega^{-1} \sqrt{\rho} \\
              &\leq \frac{1}{1+\epsilon}\sqrt{\rho} \left[(1+\epsilon) \rho^{-1}\right] \sqrt{\rho}\
               =    \1,
\end{split}\]
and similarly $Y Y^\dagger \leq \1$.
From this it follows that
\begin{align*}
  \vartheta^{A_2} = \tr_{A_1} (X^\dagger X \ox \1)\psi &\leq \psi^{A_2} = \sigma^T, \text{ and} \\
  \vartheta^{A_1} = \tr_{A_2} (\1 \ox Y^\dagger Y)\varphi &\leq \varphi^{A_1} = \rho.
\end{align*}

It remains to bound the inner product $|\bra{\psi} \vartheta\rangle|$
(the other one, $|\bra{\varphi} \vartheta\rangle|$, is completely analogous):
\[\begin{split}
  |\bra{\psi} \vartheta\rangle| 
            &= \frac{1}{\sqrt{1+\epsilon}}
                \left| \bra{\Phi}\left(\rho^{1/2}\omega^{-1/2}\sigma^{1/2} \ox \sqrt{\sigma}^T \right)\ket{\Phi} \right| \\
            &= \frac{1}{\sqrt{1+\epsilon}} \left| \tr\sqrt{\rho}\omega^{-1/2}\sigma \right|                              \\
            &= \frac{1}{\sqrt{1+\epsilon}} \left| \tr\sqrt{\rho}\omega^{-1/2}[(1+\epsilon)\omega-\epsilon\Delta] \right| \\
            &= \frac{1}{\sqrt{1+\epsilon}} \left| (1+\epsilon)\tr\sqrt{\rho}\sqrt{\omega} 
                                                          - \epsilon\tr\sqrt{\rho}\omega^{-1/2}\Delta \right|       \\
            &\geq \tr\sqrt{\rho}\sqrt{(1+\epsilon)\omega} - \epsilon \bigl| \tr X\Delta \bigr|                \\
            &\geq \tr\sqrt{\rho}\sqrt{\rho} - \epsilon \|X\|\,\|\Delta\|_1                                     
             \geq 1-\epsilon,
\end{split}\]
where we have first used the definitions of $\ket{\psi}$, $\ket{\vartheta}$
and $\ket{\Phi}$, and then the identity between $\omega$ and $\sigma$; the fifth line
is by triangle inequality, in the sixth we used $(1+\epsilon)\omega \geq \rho$
once more, the operator monotonicity of the square root, and the H\"older
inequality $|\tr X\Delta| \leq \|X\|\,\|\Delta\|_1$;
in the last step we use the fact that both $\rho$ and $\Delta$ are states
and $\|X\| \leq 1$.

Finally, to obtain $\Theta$, we write
\begin{align*}
  \rho     &= \proj{\vartheta}^{A_1} + (1-\bra{\vartheta}\vartheta\rangle)\Delta_1, \\
  \sigma^T &= \proj{\vartheta}^{A_2} + (1-\bra{\vartheta}\vartheta\rangle)\Delta_2,
\end{align*}
with bona fide states $\Delta_1$ and $\Delta_2$. 
It is straightforward to check that the definition
\[
  \Theta := \proj{\vartheta} + (1-\bra{\vartheta}\vartheta\rangle)\Delta_1\ox\Delta_2
\]
satisfies all requirements on $\Theta$.
\end{proof}

\medskip
\begin{remark}
\normalfont
Although the above proof refers to the unnormalized vector $\ket{\Phi}$,
and thus taken literally only makes sense for finite dimensional
Hilbert spaces, the proposition remains true also in the infinite
dimensional (separable) case. This can be seen either by finite
dimensional approximation, or by considering $\ket{\Phi}$ as a
formal device to mediate between normalized entangled vectors 
($\ket{\varphi}$, $\ket{\psi}$, $\ket{\vartheta}$, etc) and Hilbert-Schmidt 
class operators ($\sqrt{\rho}$, $\sqrt{\sigma}$, $\rho^{1/2}\omega^{-1/2}\sigma^{1/2}$, etc).
\end{remark}

\section{Relative entropy distances}
\label{sec:relative}
The same method employed in Lemma~\ref{lemma:alicki-fannes-new} can be 
used to derive asymptotic continuity bounds for the 
relative entropy distance with respect to any closed convex set 
$C$ of states, or more generally positive semidefinite operators, 
on a Hilbert space $A$, cf.~\cite{Synak-RadtkeHorodecki}),
\begin{equation}
  \label{eq:defi-DC}
  D_C(\rho) = \min_{\gamma\in C} D(\rho\|\gamma).
\end{equation}
Unlike~\cite{Synak-RadtkeHorodecki}, $C$ has to contain only 
at least one full-rank state, so that $D_C$ is guaranteed to be finite;
in addition, $C$ should be bounded, so that $D_C$ is bounded from below.
We recover the conditional entropy $S(A|B)_\rho$ for a bipartite
state $\rho$ on $A\ox B$, as $D_C(\rho)$ with
\[
  C = \big\{ \1^A\ox\sigma^B : \sigma \text{ a state on } B \bigr\}.
\]

\begin{lemma}
  \label{lemma:D_C}
  For a closed, convex and bounded set $C$ of positive semidefinite operators, containing at
  least one of full rank, let
  \[
    \kappa := \sup_{\tau,\tau'} D_C(\tau) - D_C(\tau')
  \]
  be the largest variation of $D_C$. Then, for any two states $\rho$ and $\sigma$
  with $\frac12 \|\rho-\sigma\|_1 \leq \epsilon$,
  \begin{equation}
    \label{eq:DC-continuity}
    |D_C(\rho)-D_C(\sigma)| \leq \epsilon\,\kappa + (1+\epsilon) \,h\!\left(\!\frac{\epsilon}{1+\epsilon}\!\right).
  \end{equation}
\end{lemma}
\begin{proof}
The only modification with respect to the proof of Lemma~\ref{lemma:alicki-fannes-new} 
is that we replace the invocation of concavity
of the conditional entropy with the joint convexity of the relative
entropy, which makes $D_C$ a convex functional.

Namely, with $\omega$ as in Eq.~(\ref{eq:decompositions}),
we have on the one hand,
\[
  D_C(\omega) \leq \frac{1}{1+\epsilon} D_C(\sigma) + \frac{\epsilon}{1+\epsilon} D_C(\Delta).
\]
On the other hand, with an optimal $\gamma\in C$,
\[\begin{split}
  D_C(\omega) &=      D\bigl(\omega\|\gamma\bigr)         \\
              &=    - S(\omega) - \tr\omega\log\gamma               \\
              &\geq - h\!\left(\!\frac{\epsilon}{1+\epsilon}\!\right)
                    - \frac{1}{1+\epsilon}S(\rho) - \frac{\epsilon}{1+\epsilon}S(\Delta') \\
              &\phantom{-h\!\left(\!\frac{\epsilon}{1+\epsilon}\!\right)}
                    - \frac{1}{1+\epsilon}\tr\rho\log\gamma - \frac{\epsilon}{1+\epsilon}\tr\Delta'\log\gamma \\
              &=    - h\!\left(\!\frac{\epsilon}{1+\epsilon}\!\right)
                    + \frac{1}{1+\epsilon}D(\rho\|\gamma) + \frac{\epsilon}{1+\epsilon}D(\Delta'\|\gamma) \\
              &\geq - h\!\left(\!\frac{\epsilon}{1+\epsilon}\!\right)
                       + \frac{1}{1+\epsilon} D_C(\rho) + \frac{\epsilon}{1+\epsilon} D_C(\Delta').
\end{split}\]
Putting these two inequalities together yields the claim of the lemma.
\end{proof}

\medskip
In particular, in the case that
\[\begin{split}
  C &= \text{SEP}(A:B) \\
    &:= \operatorname{conv}\{\alpha^A\ox\beta^B : \alpha,\,\beta \text{ states on }A,\,B,\text{ resp.} \}
\end{split}\]
is the set of separable states, we obtain the
relative entropy of entanglement of a state $\rho$ on bipartite system $A\ox B$, 
$E_R(\rho) = D_{\text{SEP}(A:B)}(\rho)$~\cite{relent}. Furthermore,
we consider its regularization
\[
  E_R^\infty(\rho) = \lim_{n\rightarrow\infty} \frac1n E_R(\rho^{\ox n}),
\]
which is known to be different from $E_R(\rho)$ in general~\cite{VollbrechtWerner}.

\begin{corollary}
  {\bf (Cf.~Donald/Horodecki~\cite{DonaldHorodecki} \&{} Christandl~\cite{Christandl:PhD})}
  \label{cor:E_R}
  For any two states $\rho$ and $\sigma$ on the composite system $A\ox B$,
  denoting the smaller of the dimensions $|A|$, $|B|$ by $d$,
  $\frac12 \|\rho-\sigma\|_1 \leq \epsilon$ implies
  \begin{align*}
    |E_R(\rho)-E_R(\sigma)| &\leq \epsilon \log d + (1+\epsilon) \,h\!\left(\!\frac{\epsilon}{1+\epsilon}\!\right), \\
    |E_R^\infty(\rho)-E_R^\infty(\sigma)| 
                            &\leq \epsilon \log d + (1+\epsilon) \,h\!\left(\!\frac{\epsilon}{1+\epsilon}\!\right).
  \end{align*}
\end{corollary}
Note that this bound only depends on the smaller of the two
dimensions, in contrast to~\cite{DonaldHorodecki}; in particular,
it applies even in the case that one of the two Hilbert spaces is infinite
dimensional.

\medskip
\begin{proof}
The first bound, on the single-letter $E_R$ is a direct application of Lemma~\ref{lemma:D_C}
to the case where $C$ is the set of all separable states on $A\ox B$.

For the regularization, consider any integer $n$ and
\[\begin{split}
  \Big| E_R&\bigl(\rho^{\ox n}\bigr) - E_R\bigl(\sigma^{\ox n}\bigr) \Big| \\
       &=     \left| \sum_{t=1}^n E_R\bigl(\rho^{\ox t} \ox \sigma^{\ox n-t}\bigr)
                                  - E_R\bigl(\rho^{\ox t-1} \ox \sigma^{\ox n-t+1}\bigr) \right| \\
       &\leq \sum_{t=1}^n | E_R(\rho \ox \Omega_t) - E_R(\sigma \ox \Omega_t) |,
\end{split}\]
with $\Omega_t = \rho^{\ox t-1} \ox \sigma^{\ox n-t}$. Now for each $t$,
Lemma~\ref{lemma:D_C} gives
\[
  | E_R(\rho \ox \Omega_t) - E_R(\sigma \ox \Omega_t) | 
                 \leq \epsilon\kappa_t + (1+\epsilon) \,h\!\left(\!\frac{\epsilon}{1+\epsilon}\!\right),
\]
with $\kappa_t = \sup_{\tau,\tau'} \bigl( E_R(\tau\ox\Omega_t) - E_R(\tau'\ox\Omega_t) \bigr)$.
To see this, we have to look into the proof of the lemma, and observe that for
states $\rho\ox\Omega_t$ and $\sigma\ox\Omega_t$, also the auxiliary
operators $\Delta$ and $\Delta'$ are of the form $\tau\ox\Omega_t$
and $\tau'\ox\Omega_t$. However, by LOCC monotonicity,
\[
  E_R(\tau' \ox \Omega_t) \geq E_R(\Omega_t), 
\]
and similarly
\[
  E_R(\tau \ox \Omega_t) \leq E_R(\Phi_d \ox \Omega_t) \leq \log d + E_R(\Omega_t),
\]
so that $\kappa_t \leq \log d$. Although we do not need it, the right
hand inequality is in fact an equality,
$E_R(\Phi_d \ox \Omega_t) = \log d + E_R(\Omega_t)$~\cite{extraweak}.
Thus, we obtain for all $n$,
\[
  \left| \frac1n E_R\bigl(\rho^{\ox n}\bigr) - \frac1n E_R\bigl(\sigma^{\ox n}\bigr) \right| \\
                       \leq \epsilon \log d + (1+\epsilon) \,h\!\left(\!\frac{\epsilon}{1+\epsilon}\!\right),
\]
and taking the limit $n\rightarrow\infty$ concludes the proof.
\end{proof}

\medskip
Again, in Lemma~\ref{lemma:D_C} and Corollary~\ref{cor:E_R}, the constant in 
the linear term (proportional to $\epsilon$) is essentially best possible, as we see by taking 
two states maximizing the difference $D_C(\rho)-D_C(\sigma)$, i.e.~attaining
$\kappa$, since $\frac12 \|\rho-\sigma\|_1 \leq 1 =:\epsilon$.

\begin{remark}
\normalfont
Lemma~\ref{lemma:D_C} improves upon similar-looking general bounds by
Synak-Radtke and Horodecki~\cite{Synak-RadtkeHorodecki}, which were 
subsequently optimized by Mosonyi and Hiai~\cite[Prop.~VI.1]{MosonyiHiai}. 
The latter paper also explains lucidly (in Sec.~VI) that the coefficient 
$\frac{1}{1+\epsilon}$ in the convex decomposition of $\omega$ 
in two ways, into $\rho$ and $\Delta'$ and into $\sigma$ and $\Delta$, 
is optimal, and gives a nice geometric interpretation of $\omega$ as 
a $\max$-relative entropy center of $\rho$ and $\sigma$ (cf.~\cite{Kimura-et-al}).
Thus, at least following the same strategy one cannot improve the bound any more.

That the regularized relative entropy measure $E_R^\infty$ is asymptotically 
continuous followed previously from its non-lockability~\cite{HHHO-lock},
which it inherits from $E_R$. 
This has been worked out in~\cite[Prop.~13]{antisymm}, 
following~\cite[Prop.~3.23]{Christandl:PhD}, with a different linear term.
\end{remark}

\begin{remark}
\normalfont
It would be interesting to lift the restriction that $C$ has to be
a convex set: Natural examples are the case that $C$ is the set of all
product states in a bipartite (multipartite) system, in which case
$D_C$ becomes the quantum mutual information (multi-information);
or the case that $C$ is the closure of the set of all Gibbs states for a suitable
Hamiltonian operator $H$,
\[
  C = \overline{\left\{ \frac{1}{\tr e^{-\beta H}} e^{-\beta H} : \beta > 0 \right\}}.
\]

Both examples have in common that $C$ is an exponential family (or the closure
of one); it is known that at least in some cases $D_C$ is continuous, but 
counterexamples of discontinuous behaviour are known~\cite{WeisKnauf}.
\end{remark}

\section{Bounded energy}
\label{sec:bounded}
If the Hilbert space in the Fannes inequality (Lemma~\ref{lemma:fannes-audenaert})
has infinite dimension, or likewise $A$ in the Alicki-Fannes inequality
(Lemma~\ref{lemma:alicki-fannes-new}), then the bound becomes trivial:
the right hand side is infinite. This is completely natural, since the
entropy is not even continuous, and these Fannes-type bounds 
imply a sort of uniform continuity. Continuity is restored, however, when
restricting to states of finite energy, for instance of a quantum harmonic
oscillator~\cite{Wehrl}, see also~\cite{ESP} and~\cite{Shirokov:S-continuity}
for more recent results and excellent surveys on the status of continuity of 
the entropy. Shirokov~\cite{Shirokov:I-continuity} has developed an 
approach do prove (local) continuity of entropic quantities, based on certain 
finite entropy assumptions, in which he uses Alicki-Fannes inequalities
on finite approximations.

Uniform bounds are still out of the question, but what we shall show here 
is that the Fannes and Alicki-Fannes inequalities discussed above have 
satisfying analogues, with a dependence on the energy of 
the states rather than the Hilbert space dimension.

Abstractly, our setting is this:
Consider a Hamiltonian $H$ on a infinite dimensional 
separable Hilbert space $A$. 
If there is another system $B$ and we consider bipartite states
and conditional entropy, we implicitly assume trivial Hamiltonian
on $B$, i.e.~global Hamiltonian $H = H^A\ox\1^B$.
We shall need a number of assumptions on $H$, to start with that 
it has discrete spectrum and that it is bounded from below; for normalization
purposes we fix the ground state energy of $H$ to be $0$. The mathematically
precise assumption is the following.

\medskip\noindent
{\bf Gibbs Hypothesis.} For every $\beta > 0$, let the partition function
$Z(\beta) := \tr e^{-\beta H} < \infty$ be finite, 
so that $\frac{1}{Z(\beta)} e^{-\beta H}$ is a bona fide state
with finite entropy. In this case, for every energy $E$ in the spectrum of $H$,
the (unique) maximizer of the entropy $S(\rho)$ subject to $\tr\rho H\leq E$
is of the Gibbs form:
\[
  \gamma(E) = \frac{1}{Z(\beta(E))} e^{-\beta(E) H},
\]
where $\beta=\beta(E)$ is decreasing with $E$ and is the solution to the equation
\[
  \tr e^{-\beta H}(H-E) = 0.
\]
The entropy in this case is given by
\[
  S\bigl(\gamma(E)\bigr) = \log Z + \beta(E)(\log e)E.
\]
This implies that the spectrum is unbounded above, and that the energy
levels cannot become ``too dense'' with growing energy value.

Let us immediately draw some conclusions from these assumptions;
the following is a simply consequence of Shirokov's~\cite[Prop.~1]{Shirokov:Entropy}, 
for which we present an elementary proof.
\begin{proposition}
  \label{prop:concavity}
  For a Hamiltonian $H$ satisfying the Gibbs Hypothesis, $S\bigl(\gamma(E)\bigr)$ is
  a strictly increasing, strictly concave function of the energy $E$.
\end{proposition}
\begin{proof}
It is clear from the maximum entropy characterization of $\gamma(E)$
that the entropy as a function of $E$ must be non-decreasing; it is
unbounded by looking at the formula for the entropy in terms of $\log Z$.

Furthermore, for energies $E_1$ and $E_2$, and $0\leq p \leq 1$,
\[
  \tr\bigl(p\gamma(E_1)+(1-p)\gamma(E_2)\bigr)H \leq p E_1 + (1-p) E_2 =:E,
\]
and so concavity follows:
\begin{equation}\begin{split}
  \label{eq:concave}
  S\bigl(\gamma(E)\bigr) &\geq S\bigl(p\gamma(E_1)+(1-p)\gamma(E_2)\bigr) \\
                         &\geq p S\bigl(\gamma(E_1)\bigr) + (1-p) S\bigl(\gamma(E_2)\bigr).
\end{split}\end{equation}

From this it follows that $S\bigl(\gamma(E)\bigr)$ is strictly increasing,
because otherwise $S\bigl(\gamma(E_1)\bigr) = S\bigl(\gamma(E_2)\bigr)$
for some $E_1 < E_2$, but then $S\bigl(\gamma(E_2)\bigr) < S\bigl(\gamma(E_3)\bigr)$
for some $E_2 < E_3$, since the entropy grows to infinity as $E\rightarrow\infty$,
contradicting concavity.

But this means that for $E_1 \neq E_2$, necessarily $\gamma(E_1) \neq \gamma(E_2)$,
and so by the strict concavity of the von Neumann entropy, we have strict
inequality in the second line of Eq.~(\ref{eq:concave}) for $0<p<1$.
\end{proof}

\begin{corollary}
  \label{cor:bound}
  If $H$ satisfies the Gibbs Hypothesis, then for any $\delta>0$,
  \[
    \sup_{0<\lambda\leq\delta} \lambda\, S\bigl(\gamma(E/\lambda)\bigr) = \delta\, S\bigl(\gamma(E/\delta)\bigr).
  \]
\end{corollary}
\begin{proof}
The right hand side is clearly attained by letting $\lambda=\delta$. 
To prove ``$\leq$'' for any admissible $\lambda$, observe that
by concavity (Proposition~\ref{prop:concavity}),
\[
  S\bigl(\gamma(tF)\bigr) \geq t\,S\bigl(\gamma(F)\bigr) + (1-t)\,S\bigl(\gamma(0)\bigr)
                          \geq t\,S\bigl(\gamma(F)\bigr).
\]
Letting $t=\frac{\lambda}{\delta} \leq 1$ and $F=\frac{E}{\lambda}$
concludes the proof.
\end{proof}

\medskip
\begin{remark}
  \label{rem:SE-vs-E}
  \normalfont
  Another useful fact proved by Shirokov~\cite[Prop.~1(ii)]{Shirokov:Entropy}, which
  we shall invoke later, is that under our assumptions,
  $S\bigl(\gamma(E)\bigr) = o(E)$, which can be recast as
  saying that $\delta\, S\bigl(\gamma(E/\delta)\bigr) \rightarrow 0$
  for every finite $E$ and $\delta\rightarrow 0$. 
  \qed
\end{remark}

\medskip
We start with an easy-to-prove continuity bound for the entropy,
inspired by the proof of Lemma~\ref{lemma:fannes-audenaert},
though for the conditional entropy we shall have to resort to a
different argument.
It uses a quantum coupling as in Proposition~\ref{prop:quantum-coupling}
(which implies a weaker bound in the following, with the square 
of the expression on the right hand side).

\begin{proposition}
  \label{prop:almost-purification}
  Let $\rho$ and $\sigma$ be states on the same Hilbert space $A$, and
  consider the tensor square $A\ox A =: A_1A_2$ of the quantum system.
  Then, there exists a state $\omega$ with
  $\omega^{A_1} = \rho$, $\omega^{A_2} = \sigma$ and such that
  \[
    \| \omega \|_\infty \geq 1-\frac12 \|\rho-\sigma\|_1.
  \]
\end{proposition}

\medskip
\begin{proof}
Choose spectral decompositions 
\begin{align*}
  \rho   &= \sum_i r_i \proj{e_i}, \\
  \sigma &= \sum_i s_i \proj{f_i},
\end{align*}
of the two states, with $r_1 \geq r_2 \geq \ldots$ and $s_1 \geq s_2 \geq \ldots$;
then, the $\ell^1$-distance between the probability vectors
$(r_i)$ and $(s_i)$ is not larger than the trace distance between
$\rho$ and $\sigma$:
\[
  \|\rho-\sigma\|_1 \geq \|(r_i)-(s_i)\|_1 =: 2\epsilon.
\]
(This is known as Mirksy's inequality~\cite[Cor.~7.4.9.3]{HornJohnson}.)

Defining a vector
\[
  \ket{\phi} := \sum_i \sqrt{\min\{r_i,s_i\}} \ket{e_i}^{A_1}\ket{f_i}^{A_2}
\]
in $A_1A_2$, we clearly have $\tr\proj{\phi} = 1-\epsilon$,
and $\phi^{A_1} \leq \rho$, $\phi^{A_2} \leq \sigma$, thus we can
write
\[
  \rho = \proj{\phi}^{A_1} + \epsilon\Delta_1,
  \quad
  \sigma = \proj{\phi}^{A_2} + \epsilon\Delta_2,
\]
with bona fide states $\Delta_1$ and $\Delta_2$. 

It is straightforward to check that the definition
$\omega := \proj{\phi} + \epsilon\Delta_1\ox\Delta_2$
satisfies all requirements on $\omega$.
\end{proof}

\begin{lemma}
  \label{lemma:easy-S}
  Let the Hamiltonian $H$ on $A$ satisfying the Gibbs Hypothesis. Then
  for any two states $\rho$ and $\sigma$ on $A$ with
  $\tr \rho H,\,\tr \sigma H \leq E$ and
  $\frac12 \|\rho-\sigma\|_1 \leq \epsilon \leq 1$, 
  \[
    \bigl| S(\rho)-S(\sigma) \bigr| \leq 2\epsilon S\bigl( \gamma(E/\epsilon) \bigr) + h(\epsilon).
  \]
\end{lemma}
\begin{proof}
Pick a state $\omega$ on $A_1A_2$, according to Proposition~\ref{prop:almost-purification}:
$\omega^{A_1}=\rho$, $\omega^{A_2} = \sigma$, and with largest
eigenvalue $\geq 1-\epsilon$, meaning that we can write
\[
  \omega = (1-\epsilon)\proj{\psi} + \epsilon\omega',
\]
with a pure state $\ket{\psi}$ (the normalized vector $\ket{\phi}$
from the proof of Proposition~\ref{prop:almost-purification})
and some other state $\omega'$. Hence,
\[\begin{split}
  \bigl| S(\rho)-S(\sigma) \bigr| &=    \bigl| S(\omega^{A_1})-S(\omega^{A_2}) \bigr| \\
                                  &\leq S(\omega^{A_1A_2})                          \\
                                  &\leq \epsilon S(\omega') + h(\epsilon)           \\
                                  &\leq 2\epsilon S\bigl( \gamma(E/\epsilon) \bigr) + h(\epsilon).
\end{split}\]
Here, we have first used the marginals of $\omega$, then in the second line
the Araki-Lieb ``triangle'' inequality~\cite{ArakiLieb}, 
in the third line strong subadditivity,
and in the last step the maximum entropy principle, noting that 
with respect to the Hamiltonian $H^{A_1}\ox\1^{A_2} + \1^{A_1}\ox H^{A_2}$, 
$\omega$ has energy $\leq 2E$, and so the energy of $\omega'$ is bounded
by $2E/\epsilon$. For the last line, observe that the Gibbs state at energy 
$2E/\epsilon$ of the composite system is $\gamma(E/\epsilon)^{\ox 2}$. 
\end{proof}

\medskip
The following two general bounds lack perhaps the simple elegance of
Lemma~\ref{lemma:easy-S}, but they turn out to be more flexible,
and stronger in certain regimes.

\begin{meta-lemma}[Entropy]
  \label{lemma:meta-S}
  For a Hamiltonian $H$ on $A$ satisfying the Gibbs Hypothesis and any
  two states $\rho$ and $\sigma$ with
  $\tr \rho H,\,\tr \sigma H \leq E$,
  $\frac12 \|\rho-\sigma\|_1 \leq \epsilon < \epsilon' \leq 1$, 
  and $\delta = \frac{\epsilon'-\epsilon}{1+\epsilon'}$,
  \[
    \bigl| S(\rho) \!-\! S(\sigma) \bigr| 
                 \leq (\epsilon'+2\delta) S\bigl( \gamma(E/\delta) \bigr) + h(\epsilon') + h(\delta).
  \]
\end{meta-lemma}

\begin{meta-lemma}[Conditional entropy]
  \label{lemma:meta-S-cond}
  For states $\rho$ and $\sigma$ on the bipartite system $A\ox B$
  and otherwise the same assumption as before,
  \[\begin{split}
    \bigl| S(A|B)_\rho \!-\! S(A|B)_\sigma \bigr| 
                 &\leq (2\epsilon'+4\delta) S\bigl( \gamma(E/\delta) \bigr) \\
                 &\phantom{==}
                       + (1+\epsilon')\,h\!\left(\!\frac{\epsilon'}{1+\epsilon'}\!\right) + 2 h(\delta).
  \end{split}\]
\end{meta-lemma}

To interpret these bounds, we remark that in a certain sense they show
that the Gibbs entropy at the cutoff energy $E/\epsilon$ ($E/\delta$) 
takes on the role of the logarithm of the dimension in the finite dimensional case.
Before we launch into their proof, let us introduce 
some notation: Define the energy cutoff projectors
\[
  P_\leq := \sum_{0\leq E_n \leq E/\delta} \proj{n},
  \quad
  P_> := \1-P_\leq,
\]
where $\ket{n}$ is the eigenvector of eigenvalue $E_n$ of the Hamiltonian $H$. 
We shall also consider the pinching map
\[
  \mathcal{T}(\xi) = P_\leq \xi P_\leq + P_> \xi P_>,
\]
which is a unital channel, as well as its action on the original
$\rho$ and $\sigma$:
\begin{align*}
  \mathcal{T}(\rho)   &=: (1-\lambda)\rho_\leq + \lambda\rho_>, \\
  \mathcal{T}(\sigma) &=: (1-\mu)\sigma_\leq   + \mu\sigma_>.
\end{align*}
Note that because $H$ commutes with the action of $\mathcal{T}$,
we have $\tr\xi H = \tr\mathcal{T}(\xi)H$, and so the energy
bound $E$ applies also to $\mathcal{T}(\rho)$ and $\mathcal{T}(\sigma)$.
Hence,
\begin{equation}
  \label{eq:cutoff-bounds}
  \lambda \leq \delta,\ \lambda \tr \rho_> H \leq E,
  \quad
  \mu \leq \delta,\ \mu \tr \sigma_> H \leq E.
\end{equation}
Our strategy will be to relate $S(\rho)$ to $S(\rho_\leq)$ 
(and the same for $\sigma$ and $\sigma_\leq$) via entropy inequalities, 
including concavity, similar to the first part of the paper,
and then apply the usual Fannes (Alicki-Fannes) inequalities to 
$\rho_\leq$ and $\sigma_\leq$.

\medskip
\begin{proofthm}{of Lemma~\ref{lemma:meta-S}}
First of all, by concavity of the entropy (monotonicity under unital
cptp maps),
\begin{equation}\begin{split}
  \label{eq:rho-pre-upper}
  S(\rho) &\leq S(\mathcal{T}(\rho))                                      \\
          &=    h(\lambda) + (1-\lambda) S(\rho_\leq) + \lambda S(\rho_>).
\end{split}\end{equation}
Now, by Eq.~(\ref{eq:cutoff-bounds}), the maximum entropy principle
and Corollary~\ref{cor:bound},
\[
  \lambda S(\rho_>) \leq \lambda S\bigl( \gamma(E/\lambda) \bigr) 
                    \leq \delta S\bigl( \gamma(E/\delta) \bigr).
\]
Thus, from Eq.~(\ref{eq:rho-pre-upper}), observing $\delta \leq \frac12$, we get
\begin{equation}
  \label{eq:rho-upper}
  S(\rho) \leq S(\rho_\leq) + h(\delta) + \delta S\bigl( \gamma(E/\delta) \bigr),
\end{equation}
and likewise for $\sigma$.

Second, we have 
\begin{equation}
  \label{eq:sigma-pre-lower}
  S(\sigma) \geq (1-\mu) S(\sigma_\leq) + \mu S(\sigma_>).
\end{equation}
To see this, we think of the action of $\mathcal{T}$ as a binary measurement
on the system $A$, which we can implement coherently
with two ancilla qubits $X$ and $X'$,
\[
  \ket{\varphi} \longmapsto (P_\leq\ket{\varphi})^A\ket{00}^{XX'} + (P_>\ket{\varphi})^A\ket{11}^{XX'}.
\]
Applying this to $\sigma$, we have by unitary invariance and
the Araki-Lieb ``triangle'' inequality,
\[\begin{split}
  S(\sigma) = S(AXX') &\geq S(AX)-S(X') \\
                      &=    S(AX)-S(X)  \\
                      &=    S(A|X)      
                       =    (1-\mu) S(\sigma_\leq) + \mu S(\sigma_>).
\end{split}\]
Thus, using that the energy of $\sigma_\leq$ is at most $E/\delta$ by construction,
and so $S(\sigma_\leq) \leq S\bigl( \gamma(E/\delta) \bigr)$,
\begin{equation}
  \label{eq:sigma-lower}
  S(\sigma) \geq (1-\mu) S(\sigma_\leq) 
            \geq S(\sigma_\leq) - \delta S\bigl( \gamma(E/\delta) \bigr).
\end{equation}

Third, by definitions, contractivity of the trace norm and triangle inequality,
\[\begin{split}
  2\epsilon &\geq \|\rho-\sigma\|_1 \\
            &\geq \bigl\| P_\leq\rho P_\leq - P_\leq\sigma P_\leq \bigr\|_1    \\
            &=    \bigl\| (1-\lambda)\rho_\leq - (1-\mu)\sigma_\leq \bigr\|_1  \\
            &=    \bigl\| (1-\delta)(\rho_\leq - \sigma_\leq)
                           + (\delta-\lambda)\rho_\leq + (\mu-\delta)\sigma_\leq \bigr\|_1 \\
            &\geq (1-\delta) \|\rho_\leq-\sigma_\leq\|_1 - 2\delta,
\end{split}\]
and so
\begin{equation}
  \label{eq:trace-norm-truncated}
  \frac12 \|\rho_\leq-\sigma_\leq\|_1 \leq \frac{\epsilon+\delta}{1-\delta} = \epsilon'.
\end{equation}
Hence by the Fannes inequality in the form of Lemma~\ref{lemma:fannes-audenaert},
\begin{equation}\begin{split}
  \label{eq:easy-part}
  |S(\rho_\leq)-S(\sigma_\leq)| &\leq \epsilon' \log\tr P_{\leq} + h(\epsilon') \\
                                &\leq \epsilon' S\bigl( \gamma(E/\delta) \bigr) + h(\epsilon').
\end{split}\end{equation}
The latter inequality holds because the state $\frac{1}{\tr P_\leq} P_\leq$ clearly
has energy bounded by $E/\delta$, and so cannot have entropy larger than the 
Gibbs state.

With these three elements we can conclude the proof: 
W.l.o.g.~$S(\rho) \geq S(\sigma)$, and so from Eqs.~(\ref{eq:rho-upper}),
(\ref{eq:sigma-lower}) and (\ref{eq:easy-part}),
\[\begin{split}
  S(\rho)-S(\sigma) &\leq S(\rho_\leq) - S(\sigma_\leq)
                           + h(\delta) + 2\delta S\bigl( \gamma(E/\delta) \bigr) \\
                    &\leq (\epsilon' + 2\delta) S\bigl( \gamma(E/\delta) \bigr) + h(\epsilon') + h(\delta),
\end{split}\]
as advertised.
\end{proofthm}

\medskip
\begin{proofthm}{of Lemma~\ref{lemma:meta-S-cond}}
It is very similar to the previous one, only that we have to be a bit more 
careful in some details, as the conditional entropy can be negative.

The first step goes through almost unchanged, with the map
$\mathcal{T}\ox\id_B$, since the conditional entropy is concave as well
(equivalent to strong subadditivity)~\cite{SSA}:
\begin{equation*}\begin{split}
  S(A|B)_\rho &\leq S(A|B)_{\mathcal{T}(\rho)}                                     \\
              &=    h(\lambda) + (1-\lambda) S(A|B)_{\rho_\leq} + \lambda S(A|B)_{\rho_>}.
\end{split}\end{equation*}
The remainder term $\lambda S(A|B)_{\rho_>}$ is upper bounded by $\lambda S(\rho_>^A)$
(again by strong subadditivity), hence the upper bound $\lambda S\bigl( \gamma(E/\lambda) \bigr)$
still applies. The only change is due to the fact that the conditional
entropy can be negative. However, for any bipartite state $\xi^{AB}$,
\[
  -S(\xi^A) \leq S(A|B)_\xi \leq S(\xi^A).
\]
Here, the right hand inequality is strong subadditivity that we have used before;
introducing a purification $\ket{\varphi}^{ABC}$ of the state, we have
$-S(A|B)_\varphi = S(A|C)_\varphi \leq S(\xi^A)$, which is the left
hand inequality.
Thus, 
\[
  (1-\lambda)S(A|B)_{\rho_\leq} \leq S(A|B)_{\rho_\leq} + \delta S\bigl( \gamma(E/\delta) \bigr).
\]
Altogether,
\begin{equation}
  \label{eq:rho-cond-upper}
  S(A|B)_\rho \leq S(A|B)_{\rho_\leq} + 2\delta S\bigl( \gamma(E/\delta) \bigr) + h(\delta) .
\end{equation}

Also the second step requires only minor modifications:
With the notation of the previous proof, and using the Araki-Lieb 
``triangle'' inequality once again,
\[\begin{split}
  S(AX & X'|B) =    S(AXX'B)-S(B) \\
              &\geq S(ABX)-S(X')-S(B) \\
              &=    S(ABX)-S(BX)-S(X)-S(B)+S(XB) \\ 
              &=    S(A|BX) - I(X:B) \\
              &\geq S(A|BX) - h(\delta).
\end{split}\]
Again, since conditional entropies can be negative, we have to
be more careful with remainder terms and get
\begin{equation}
  \label{eq:sigma-cond-lower}
  S(A|B)_\sigma \geq S(A|B)_{\sigma_\leq} - 2\delta S\bigl( \gamma(E/\delta) \bigr) - h(\delta).
\end{equation}

In the third step, the trace norm estimate (\ref{eq:trace-norm-truncated}) 
goes through unchanged, and then we apply the Alicki-Fannes inequality in the
form of Lemma~\ref{lemma:alicki-fannes-new}:
\[\begin{split}
  \bigl| S(A|B)_{\rho_\leq} &- S(A|B)_{\sigma_\leq} \bigr| \\
                            &\leq 2\epsilon' \log\tr P_{\leq}
                                   + (1+\epsilon')\,h\left(\!\frac{\epsilon'}{1+\epsilon'}\!\right) \\
                            &\leq 2\epsilon' S\bigl( \gamma(E/\delta) \bigr) 
                                   + (1+\epsilon')\,h\left(\!\frac{\epsilon'}{1+\epsilon'}\!\right).
\end{split}\]

Putting this together with Eqs.~(\ref{eq:rho-cond-upper}) and (\ref{eq:sigma-cond-lower}), 
assuming w.l.o.g.~that $S(A|B)_\rho \geq S(A|B)_\sigma$, we obtain
\[\begin{split}
  S(A|B)_\rho-S(A|B)_\sigma &\leq S(A|B)_{\rho_\leq} - S(A|B)_{\sigma_\leq}  \\
                            &\phantom{==}
                                    + 2 h(\delta) + 4\delta S\bigl( \gamma(E/\delta) \bigr) \\
                            &\leq (2\epsilon' + 4\delta) S\bigl( \gamma(E/\delta) \bigr)               \\
                            &\phantom{==}
                                    + (1+\epsilon')\,h\left(\!\frac{\epsilon'}{1+\epsilon'}\!\right) 
                                    + 2 h(\delta),
\end{split}\]
and we are done.
\end{proofthm}

\medskip
The bounds of Lemmas~\ref{lemma:easy-S}, \ref{lemma:meta-S} and \ref{lemma:meta-S-cond}
are very general, and it may not be immediately apparent how useful they are.
However, thanks to~\cite[Prop.~1(ii)]{Shirokov:Entropy}, restated in Remark~\ref{rem:SE-vs-E},
$\delta S\bigl( \gamma(E/\delta) \bigr) \rightarrow 0$ for every finite $E$,
as $\delta \rightarrow 0$ (cf.~\cite[Cor. ~4]{Shirokov:squashed}).
Thus, choosing $\epsilon'=\sqrt{\epsilon}$, the lemmas do prove continuity 
of the entropy and conditional entropy in general, 
and uniformly for each fixed energy.

We now specialize our bounds to the important case of a collection of $\ell$ 
quantum harmonic oscillators, where we shall see that the bounds are 
asymptotically tight. The Hamiltonian is
\begin{equation}
  \label{eq:harmonic}
  H = \sum_{i=1}^\ell \hbar\omega_i \, a_i^\dagger a_i,
\end{equation}
where $\omega_i$ is the native frequency of the $i$-th oscillator
and $a_i$ is its annihilation (aka lowering) operator (see e.g.~\cite{KokLovett}
or \cite{Weedbrook-rev}). 
Note that we chose the slightly unusual energy convention such that the ground state
has energy $0$, rather than $\sum_i \frac12 \hbar\omega_i$, to be able to
apply directly our above results. 
In the case of a single mode, and choosing units such that $\hbar\omega_1 = 1$,
the Hamiltonian simply becomes the number operator $N$.
In that case, it is well-known that 
\[\begin{split}
  S\bigl(\gamma(N)\bigr) = g(N) &:=   (N+1)\log(N+1) - N\log N \\
                                                  &\leq \log(N+1) + \log e.
\end{split}\]
Crucially, and in accordance with Proposition~\ref{prop:concavity}, $g$ is a 
concave, monotone increasing function of $N$.

In the general case of Eq.~(\ref{eq:harmonic}), 
$\gamma(E) = \bigotimes_{i=1}^\ell \gamma_i(E_i)$, with $E=\sum_i E_i$
and where $\gamma_i(E_i)$ is the Gibbs state of the $i$-th mode with
energy $E_i$. Maximizing the entropy,
\[
  S\left( \bigotimes_{i=1}^\ell \gamma_i(E_i) \right) 
             = \sum_{i=1}^\ell \, g\!\left(\frac{E_i}{\hbar\omega_i}\right),
\]
over all allocations of the total energy over the $\ell$ modes
leads to a transcendental equation, but we do not need to solve it as we
only want an upper bound, via $g(N) \leq \log(N+1)+\log e$. 
By a straightforward Lagrange multiplier calculation we see that the 
optimum is to divide the energy equally among the modes:
\begin{equation}\begin{split}
  \label{eq:harmonic-entropy}
  S\bigl( \gamma(E) \bigr) 
         &\leq \max \sum_{i=1}^\ell \left[ \log\left(\frac{E_i}{\hbar\omega_i}+1\right) + \log e \right] \\
         &=    (\log e)\ell + \sum_{i=1}^\ell \log\left(\frac{\overline{E}}{\hbar\omega_i}+1\right),
\end{split}\end{equation}
with $E =: \ell\overline{E}$.

By using this upper bound in Lemmas~\ref{lemma:meta-S} and \ref{lemma:meta-S-cond},
for $\delta = \alpha\epsilon(1-\epsilon)$, with a parameter $\alpha$ between $0$ and $\frac12$,
and introducing
\[
  \widetilde{h}(x) := \begin{cases}
                        h(x) & \text{ for } x \leq \frac12, \\
                        1    & \text{ for } x \geq \frac12,
                      \end{cases}
\]
we obtain directly the following:
\begin{lemma}
  \label{lemma:l-osc}
  Consider two states $\rho$ and $\sigma$ of the $\ell$-oscillator system
  (\ref{eq:harmonic}), whose energies are bounded $\tr\rho H,\ \tr\sigma H \leq E = \ell\overline{E}$.
  Then, $\frac12 \|\rho-\sigma\|_1 \leq \epsilon < 1$ implies
  \[\begin{split}
    \bigl| S&(\rho) - S(\sigma) \bigr| \\
           &\leq \epsilon \left(\!\frac{1+\alpha}{1-\alpha}\!+\!2\alpha\!\right) \!
                        \left[ \sum_{i=1}^\ell \log\left(\frac{\overline{E}}{\hbar\omega_i}+1\right) 
                                                      + \ell \log\frac{e}{\alpha(1-\epsilon)} \right]   \\
           &\phantom{==========} 
               + (\ell+2)\left(\!\frac{1+\alpha}{1-\alpha}\!+\!2\alpha\!\right)
                         \widetilde{h}\left(\frac{1+\alpha}{1-\alpha}\epsilon\right)\!. 
  \end{split}\]
  If the states live on a system composed of the $\ell$ oscillators ($A$)
  and another system $B$, then
  \[\begin{split}
    \bigl| S&(A|B)_\rho - S(A|B)_\sigma \bigr| \\
            &\leq 2\epsilon \left(\!\frac{1+\alpha}{1-\alpha}\!+\!2\alpha\!\right) \!
                        \left[ \sum_{i=1}^\ell \log\left(\frac{\overline{E}}{\hbar\omega_i}+1\right) 
                                                      + \ell \log\frac{e}{\alpha(1-\epsilon)} \right]   \\
           &\phantom{=========} 
               + (2\ell+4)\left(\!\frac{1+\alpha}{1-\alpha}\!+\!2\alpha\!\right)
                         \widetilde{h}\left(\frac{1+\alpha}{1-\alpha}\epsilon\right)\!.
             \ \Box
  \end{split}\]
\end{lemma}

\medskip
\begin{remark}
  \normalfont
  For each fixed $\epsilon \leq 1$, we can make $\alpha$ arbitrarily small,
  and then for large energy $E \gg \sum_i \hbar\omega_i$, the bounds
  of Lemma~\ref{lemma:l-osc} are asymptotically tight, in the sense that apart from
  the additive offset terms, the factor multiplying $\epsilon$ ($2\epsilon$, resp.) 
  cannot be smaller than
  \[
    S\bigl( \gamma(E) \bigr) \approx \sum_{i=1}^\ell \log\left(\frac{\overline{E}}{\hbar\omega_i}+1\right).
  \]
  This can be seen in the entropy case by comparing the vacuum
  state $\rho = \proj{0}^{\ox\ell}$ of all $\ell$ modes with the state
  $\sigma = (1-\epsilon)\proj{0}^{\ox\ell} +\epsilon\gamma(E)$; in the
  conditional entropy case, take $\rho$ to be a purification of the Gibbs
  state $\gamma(E)$ on $A\ox B$, and 
  $\sigma = (1-\epsilon)\rho + \epsilon\gamma(E)^A\ox\tau^B$ with an
  arbitrary state $\tau$ on $B$.
\end{remark}

\section{Conclusions}
Using entropy inequalities, specifically concavity, we improved the
appearance of the Alicki-Fannes inequality for the conditional von Neumann
entropy to an almost tight form.
It would be nice to know the ultimately best form among all formulas that depend
only on the dimension of the Hilbert space and the trace distance,
but we have to leave this as an open problem.

In particular, it would be curious to find the optimal form of 
the fidelity in Proposition~\ref{prop:quantum-coupling},
\[\begin{split}
  \widetilde{F} &:= \max F(\psi,\Theta) 
                       \text{ s.t. } \Theta^{A_1}=\rho,\, \Theta^{A_2}=\psi^{A_2}\!, \\
                &\geq 1-\frac12\|\rho-\sigma\|_1,
\end{split}\]
with a fixed purification $\psi$ of $\sigma$,
and of Proposition~\ref{prop:almost-purification},
\[
  1-\frac12 \|\rho-\sigma\|_1 
     \leq \max \|\omega\|_\infty \text{ s.t. } \omega^{A_1}=\rho,\, \omega^{A_2}=\sigma,
\]
which may be regarded as quantum state analogues of the 
coupling random variables,
\[
  \frac12 \|p-q\|_1 = \min \Pr\{X\neq Y\} \text{ s.t. } X\sim p,\, Y\sim q.
\]
Furthermore, are there versions of these statements that would
allow for alternative proofs or tighter versions of 
Lemmas~\ref{lemma:alicki-fannes-new} and~\ref{lemma:meta-S-cond}
for the conditional entropy?

The same principle lead to the apparently first uniform continuity
bounds of the entropy and conditional on infinite dimensional Hilbert spaces under
a bound on the expected energy (or, for that matter, bounded expectation
of any sufficiently well-behaved Hermitian operator). In the case of 
a system of harmonic oscillators, we have seen that the bound is, in 
a certain sense, asymptotically tight, even though here we are much
farther away from a universally optimal form.

The Fannes and Alicki-Fannes inequalities already are known to
have many applications in quantum information theory. These include
the continuity of certain entanglement measures such as 
entanglement of formation~\cite{Nielsen-continuity}, 
relative entropy of entanglement~\cite{DonaldHorodecki},
squashed entanglement~\cite{E-sq} and 
conditional entanglement of mutual information~\cite{YHW}, 
and of various quantum channel capacities~\cite{LeungSmith}.
In fact, we always get explicit continuity bounds in terms of
the trace distance of the states or diamond norm distance of the
channels, respectively.
While in many applications it is of minor interest to have the optimal
form of the bound (for example when $\epsilon$ goes to $0$), 
it pays off to have a tighter bound than~\cite{AlickiFannes}
in the setting of approximately degradable channels~\cite{Sutter-et-al}. 
Indeed, this results even in new, tighter upper bounds on the quantum
capacity of very quiet depolarizing channels~\cite{Sutter-et-al}, 
by way of an extension of the methodology of~\cite{Q_ss}.

The infinite dimensional versions of these entropy bounds under
an energy constraint are awaiting applications, though it seems clear
that explicit bounds on the continuity and asymptotic continuity of entanglement 
measures~\cite{ESP}, (e.g.~for squashed entanglement since the
first posting of the present manuscript~\cite{Shirokov:squashed}) and channel 
capacities~\cite{Holevo:constrained,Holevo-e-constrained,HolevoShirokov:constrained,ShirokovHolevo}
in infinite dimension should be among the first, as well as the extension of 
approximate degradability~\cite{Sutter-et-al} to Bosonic 
channels~\cite{in-prep}.

\phantom{.}
\vspace{0.6cm}
{\bf Acknowledgments.}
Thanks to David Sutter and Volkher Scholz for stimulating discussions,
to Nihat Ay, Mil\'{a}n Mosonyi and Dong Yang for remarks on general relative 
entropy distances, to Maxim Shirokov for his many insights into entropy
and entanglement measures, in particular his keen interest in the
asymptotic continuity of entanglement cost and for spotting an error in 
an earlier version of the proof of Lemma~\ref{lemma:meta-S-cond}, and 
to Mark Wilde for comments on the history of Lemma~\ref{lemma:fannes-audenaert}.
The hospitality of the Banff International Research
Station (BIRS) during the workshop ``Beyond IID in Information Theory''
(5-10 July 2016) is gratefully acknowledged, 
where Volkher Scholz and David Sutter posed the derivation of
infinite dimensional Fannes type inequalities as an open problem,
and where the present work was initiated.

The author's work was supported by the EU (STREP ``RAQUEL''), 
the ERC (AdG ``IRQUAT''),
the Spanish MINECO (grant FIS2013-40627-P)
with the support of FEDER funds, as well as by
the Generalitat de Catalunya CIRIT, project~2014-SGR-966.

\end{document}